\documentclass{ifacconf}

\usepackage{graphicx}      
\usepackage{natbib}
\usepackage{float} 
\usepackage{amsmath}
\usepackage{amsfonts}
\usepackage{enumitem}

\usepackage{tikz}
\usetikzlibrary{shapes,arrows,positioning,calc,fit, arrows.meta}

\newcommand{\R}{\mathbb{R}}
\newcommand{\cW}{\mathcal{W}}

\newtheorem{lemma}{Lemma}

\DeclareMathOperator{\diag}{diag}

\newcommand{\T}{^\top}
\newcommand{\D}{\mathrm{D}}

\newenvironment{proof}{\par\noindent{\bf Proof.}}{\hfill$\square$\par}

\tikzset{
block/.style={
  draw, 
  fill=white!20, 
  rectangle, 
  minimum height=3em, 
  minimum width=6em
  },
sum/.style={
  draw, 
  fill=white!20, 
  circle, 
  },
input/.style={coordinate},
output/.style={coordinate},
pinstyle/.style={
  pin edge={to-,thin,black}
  }
}

\tikzstyle{neuron} = [
circle, 
minimum size=1cm,
draw=black]

\begin{document}

\begin{frontmatter}

\title{Event Disturbance Rejection: \\ A Case Study} 



\author[First,Second]{Alessandro Cecconi} 
\author[First]{Michelangelo Bin} 
\author[Second,Third]{Rodolphe Sepulchre}
\author[First]{Lorenzo Marconi}

\address[First]{DEI, University of Bologna, Italy}
\address[Second]{STADIUS, KU Leuven, Belgium}
\address[Third]{Department of Engineering, University of Cambridge, UK}

\begin{abstract}
This article introduces the problem of robust event disturbance rejection. Inspired by the design principle of linear output regulation, a control structure based on excitable systems is proposed. Unlike the linear case, contraction of the closed-loop system must be enforced through specific input signals. This induced contraction enables a steady-state analysis similar to the linear case. Thanks to the excitable nature of the systems, the focus shifts from precise trajectory tracking to the regulation of discrete events, such as spikes. The study emphasizes rejecting events rather than trajectories and demonstrates the robustness of the approach, even under mismatches between the controller and the exosystem. This work is a first step towards developing a design principle for event regulation.
\end{abstract}

\begin{keyword}
Nonlinear systems, regulation theory, contraction theory, event regulation, nonlinear oscillators 
\end{keyword}

\end{frontmatter}

\section{Introduction}
\vspace*{-1em}
This article considers the problem of \emph{event disturbance rejection} in the context of excitable systems. Unlike classical regulation theory, which seeks the exact rejection of disturbance signals (i.e. $e(t)=~0$), this paper targets robust rejection of discrete events linked to continuous dynamics. Exact trajectory regulation is often unrealistic due to unavoidable model-controller mismatches. By leveraging event-based properties of excitable systems, the proposed approach ensures robustness against such uncertainties. The resulting control structure builds on the principles of linear output regulation theory.

Excitable systems, such as neurons, clearly illustrate the interplay between continuous dynamics and discrete events. Typically, these systems rest at a stable equilibrium but exhibit significant transient deviations—referred to as \emph{spikes}—in response to stimuli over a certain threshold \citep{izhikevich2007dynamical, gerstner2014neuronal}. In this context, spikes constitute discrete events clearly distinguishable from sub-threshold responses by their stereotypical shape and amplitude. In this work, an event specifically corresponds to the occurrence of a spike in the system's output. Hence, excitable systems allow for an agile characterization of discrete events even if the underlying dynamics is continuous. For this reason, they set the ground for the present study.

Traditionally, excitability has been studied by analyzing closed nonlinear systems under constant inputs \citep{winfree1980geometry, rinzel1998analysis}. However, experimental and numerical evidence suggests that contraction properties—and thus the reliability in reproducing spiking patterns—are inherently input-dependent \citep{mainen1995reliability, brette2003reliability, kosmidis2003analysis}. Motivated by this, the present work leverages contraction theory \citep{lohmiller1998contraction, pavlov_uniform_2006} to treat excitable systems as open dynamical systems, where contractive behaviors are induced by external stimuli.
The proposed approach naturally interlaces excitability, contraction, and regulation: excitable dynamics enable input-induced contraction; contraction, in turn, ensures the existence of a well-defined steady-state behavior; regulation goals are then ensured by embedding an internal model of the disturbance generator.

We start by considering the continuous-time linear version of the classical disturbance rejection problem, where both the plant and the exosystem are linear systems, and the disturbance signal $w$ a sinusoid. 
Although neurons are not describable by linear models, the linear setting is well-studied in control, and gives considerable insight on how to approach this novel event disturbance rejection problem in a nonlinear setting.
In particular, the classical linear solution in terms of trajectories \citep{davison1975robust, francis1976internal} suggests that key points in the linear case are: (i) linearity and contraction of the closed-loop system imply that the steady-state trajectories are uniquely defined by $w$, and they belong to the same class of signals (they oscillate at the same frequency); (ii) as a consequence, it is necessary and sufficient that the closed-loop system embeds a suitable internal model of the process generating $w$, so as to guarantee the existence of a good steady-state; (iii) contraction enables a simplified steady-state analysis by which one can prove that all closed-loop solutions converge to the target steady-state, thereby meeting the regulation objectives.
Next, by leveraging the insight of the linear solution, a nonlinear version of the problem is considered, where the exosystem and the plant are made up of properly interconnected excitable systems. The goal is to cancel out from the plant's output a periodic disturbance signal, seen as a train of spiking events, generated by this nonlinear oscillator.
Finally, the paper explores the extension to open systems by introducing another driving signal acting on the plant. In this case, the disturbance rejection problem at hand can be interpreted as the problem of restoring an input-output behavior given by the unperturbed controlled system. Some numerical simulations illustrate, in this simple case study, the performance and robustness of the proposed control structure.

\section{The Linear Motif}
This section develops the continuous-time linear version of the considered disturbance rejection problem. The structure of the overall control system is depicted in Figure \ref{fig.control_structure_linear}. 
In this case, the disturbance signal $w$ is a sinusoid, without loss of generality of unitary angular frequency. Hence, it is generated by the linear oscillator
\begin{align}\label{s.linear.w}
	\dot x_W &= \begin{bmatrix}
		0 & 1\\
		-1 & 0
	\end{bmatrix}x_W,
	&
	w&=\begin{bmatrix}
		1 & 0
	\end{bmatrix} x_W.
\end{align}

The plant $\Sigma_p$ to be controlled is modeled as a Bounded-Input Bounded-Output (BIBO) stable linear system with inputs $w$ and $u$, and output $y$.
For the sake of illustration, we consider the simplest model with such properties, i.e.
\begin{equation}\label{s.xp}
	\begin{aligned}
		\dot x &= - x + w - u \\
		y &= x.
	\end{aligned} 
\end{equation}

The goal is to design the controller $\Sigma_{\eta}$ generating $u$ in such a way that $y(t) \to 0$ for every $w\in\cW$, where here $\cW$ is the set of outputs generated by \eqref{s.linear.w}. This is a disturbance rejection control problem where the signal to be rejected is $w$. In this linear case, the notion of event is not well-defined, althought it may be thought as a limiting case where the event corresponds to the presence of a given harmonic, and hence it is spread over the entire time axis.

A  solution to this problem  was given in the more general context of regulation by \citep{davison1975robust,francis1976internal}. Such a solution employs a controller $\Sigma_{\eta}$ having the following general form:
\begin{subequations} \label{s.linear}
	\begin{align}
		\dot x_\eta &= \Phi x_\eta + Gy \label{s.linear.xeta}\\
		u&= k_y y+ K_{\eta} x_\eta
	\end{align}
\end{subequations}
in which $\Phi$ has eigenvalues $\pm j$, $(\Phi,G)$ is controllable, and  $k_y$, and $K_{\eta}$ are chosen so that the closed-loop matrix
\begin{equation}\label{d.linear.cl_matrix}
	\begin{bmatrix}
		-k_y - 1 & -K_{\eta} \\G & \Phi 
	\end{bmatrix}
\end{equation} 
is Hurwitz. 
In the specific example considered here, a possible choice is
\begin{align}\label{s.linear.I.choice}
	\Phi &= \begin{bmatrix}
		0 & 1\\-1 & 0
	\end{bmatrix},& G&=\begin{bmatrix}
		1\\0
	\end{bmatrix},
	\\
	k_y&= 0, & K_{\eta_1}&>-K_{\eta_2}, & K_{\eta_2}&<1. \label{s.linear.gains}
\end{align}

 \begin{figure}[t]
	\centering
	\includegraphics[width=0.4\textwidth]{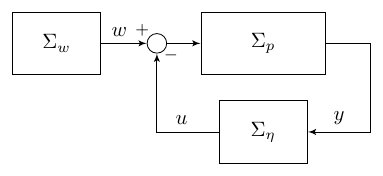}
	\caption{Block diagram of the disturbance rejection problem under consideration. The exosystem $\Sigma_w$ generates a disturbance $w$ acting on the plant $\Sigma_p$, and the control goal is to inhibit it by means of the control action $u$ generated by $\Sigma_{\eta}$.}
	\label{fig.control_structure_linear}
\end{figure}

The linear solution suggests the following key ingredients:
\begin{enumerate}[label=\textbf{I\arabic*.},ref=\textbf{I\arabic*}]
	\item\label{I1}
Closed-loop stability implies that the closed-loop system is contractive in the sense of \citep{lohmiller1998contraction} (hence convergent in the sense of \citep{pavlov_uniform_2006}). As a consequence, for each $w\in\cW$,  there is a unique steady-state solution which is asymptotically stable.
	Such a steady-state solution has the form
	\begin{equation}
		x_{cl}(t) = \Pi w(t),
		\label{eq.ss}
	\end{equation}
	where $x_{cl} = (x , x_\eta)$, and $\Pi \in \R^{2 \times 1}$ is the solution of the Francis' Equation~\citep[Ch. 4]{isidori2017lectures}. This underlines that all steady-state solutions are of the same family of the driving signal $w$.
	
	\item\label{I2} When $y=0$ as desired, the controller's equations boil down to
	\[
	\begin{aligned}
		\dot x_\eta &= \begin{bmatrix}
			0 & 1\\
			-1 & 0
		\end{bmatrix}x_\eta,\\
		w &= u = K_{\eta} x_\eta.
	\end{aligned}
	\]
	Namely, the controller embeds a copy of the exosystem generating $w$.
	
	\item\label{I3} The previous two points suggest that system $\Sigma_{\eta}$ must be designed in such a way that:
		\begin{itemize}
			\item (Stability requirement) The closed-loop system is contractive/convergent, so as all solutions converge to a unique steady-state trajectory $x_{cl}$.
			\item (Steady-state requirement) There exists a steady-state solution $x_{cl}$ of the closed-loop system in which $y=0$.
		\end{itemize}
\end{enumerate}


As the steady-state solution is unique, the above points suffice to show that every closed-loop solution satisfies $\lim_{t\to\infty}y(t)=0$.
The following section constructs a nonlinear counterpart of the linear framework considered here, in which the controller and the exosystem are composed of neuron models. 
Such an extension is constructed along the lines of \ref{I1}--\ref{I3}. 

\section{The nonlinear motif}\label{sec.nl}

This section mimics the linear approach in a framework where the systems consist of neuronal models described by equations of the following kind
\begin{subequations}\label{eq.fhn}
	\begin{align}
		\dot{x}_1 &= x_1 - \frac{1}{3}x^3_1 - x_2 +   k  \left(\sum_{i\in \mathcal{I}}\sigma_i z_i -y\right)  \\
		\dot{x}_2 &= \tau(x_1 - x_2),\\[1ex]
		y  &= x_1 ,
	\end{align} 
\end{subequations} 
where $\mathcal{I}$ is an index set indexing the presynaptic couplings, signals $z_i$  are the presynaptic voltage inputs to the neuron,  $\tau>0$ is the time constant, $k>0$ plays the role of the synaptic conductance density, and $\sigma_i\in\{-1,1\}$ satisfies $\sigma_i=1$ if the input is excitatory and $\sigma_i=-1$ if it is inhibitory. 

Model \eqref{eq.fhn} is FitzHugh-Nagumo (FN) Equation, a simplified version of the more general Hodgkin–Huxley model \citep{Murray2002}, that is sufficiently simple to enable analysis and sufficiently rich to exhibit excitable behavior.  For the sake of illustration, in \eqref{eq.fhn} neuron interconnections are modeled as a particular type of electrical synapse where the sum of the presynaptic inputs enters within a single gap junction. This considerably simplifies the contraction analysis.

As it turns out, \eqref{eq.fhn} is a convergent system.
\begin{lemma}\label{lem.contractive_FN}
	If $k > 1$, system \eqref{eq.fhn} is convergent. 
\end{lemma}
\begin{proof}
	The Jacobian of the vector field $f$ describing \eqref{eq.fhn} is
	\begin{equation*}
		\D f(x) = \begin{bmatrix}
			1-k - x_1^2 & -1\\
			\tau & -\tau 
		\end{bmatrix}
	\end{equation*}
	The matrix $P=\diag(1,1/\tau)$ satisfies, for all $x\in\R^2$, 
	\begin{equation}
		P\D f(x) +\D f(x)\T P  \le  2\begin{bmatrix}
		1-k & 0\\
		0 & -1
		\end{bmatrix}.
        \label{lem.contr.eq}
	\end{equation}
	The right-hand-side is negative definite if $k>1$.
	 Thus, the result follows from \citep[Thm. 2.29]{pavlov_uniform_2006}.
\end{proof}

\begin{figure}
    \centering
    \includegraphics[width=0.425\linewidth]{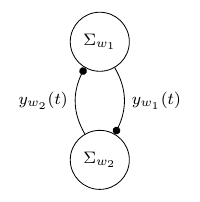}
    \caption{Nonlinear oscillator made up of two FitzHugh-Nagumo models $\Sigma_{w_1}, \Sigma_{w_2}$, interconnected via diffusive inhibitory couplings.}
    \label{fig:e_i_osc}
\end{figure}

\subsection{Neuronal Oscillator Models}
The different actors of the nonlinear neuronal system are defined as follows.
By following the linear motif, the exogenous input $w$ is modeled as the output of an oscillator. The neuronal equivalent of the harmonic oscillator \eqref{s.linear.w}  is a system defined by two neurons $\Sigma_{w_1}$ and $\Sigma_{w_2}$ linked by inhibitory couplings as described by the equations below (cf. \eqref{eq.fhn}) and depicted in Figure \ref{fig:e_i_osc}:
\begin{equation}	\label{s.w}
	\begin{aligned}
		\dot{x}_{w_1} &= x_{w_1} - \frac{1}{3}x^3_{w_1} - x_{w_2} + k_w(-y_{w_2} -y_{w_1})  \\
		\dot{x}_{w_2} &= \tau_w(x_{w_1} - x_{w_2}),\\ 
		\dot{x}_{w_3} &= x_{w_3} - \frac{1}{3}x^3_{w_3} - x_{w_4} + k_w(-y_{w_1} -y_{w_2})  \\
		\dot{x}_{w_4} &= \tau_w(x_{w_3} - x_{w_4}),\\ 
		y_{w_1}  &= x_{w_1} ,\quad 
		y_{w_2}   = x_{w_3} ,\quad 
		w  = y_{w_1},
	\end{aligned} 
\end{equation} 
where $k_w>\frac{1}{2}$, so as to guarantee that $y_{w1}$ and $y_{w2}$ asymptotically oscillate in antiphase, as established by the following lemma.

\begin{lemma}\label{lem.exo}
	If $k_w> \frac{1}{2}$, every solution of \eqref{s.w} satisfies $\lim_{t\to\infty}y_{w 1}(t)+y_{w 2}(t)= 0$.
\end{lemma}
\begin{proof}
    System \eqref{s.w} is an interconnection of two identical oscillators. In order to obtain asynchronous oscillations, it is sufficient that the uncoupled vector field $f$ of the single oscillator is contractive and odd in its arguments. 
    
    Introducing $\pm k_{w}y_{w_1}$, the first equation rewrites as
    \begin{equation*}
        \begin{aligned}
    \dot{x}_{w_1} &= x_{w_1} - \frac{1}{3}x^3_{w_1} - x_{w_2} - 2k_w y_{w_1} + k_w(y_{w_1} - y_{w_2})  \\
    \dot{x}_{w_2} &= \tau_w(x_{w_1} - x_{w_2}),\\ 
    \end{aligned}
    \end{equation*}
    where the term $k_w(y_{w_1} - y_{w_2})$ represents the coupling force. 
    Without the coupling, equation \eqref{lem.contr.eq} from Lemma \ref{lem.contractive_FN} becomes
    	\begin{equation*}
		P\D f(x) +\D f(x)\T P  \le  2\begin{bmatrix}
		1-2k_w & 0\\
		0 & -1
		\end{bmatrix},
	\end{equation*}
    proving contraction for $k_w > \frac{1}{2}$. 
    
    Moreover, It is easy to notice that the vector field is odd in $(x_{w_1}, x_{w_2})$. Thus, the result follows from \citep[Thm. 4]{wang2005partial}.
\end{proof}

\begin{figure}[h]
	\centering
	\includegraphics[width=0.4\textwidth]{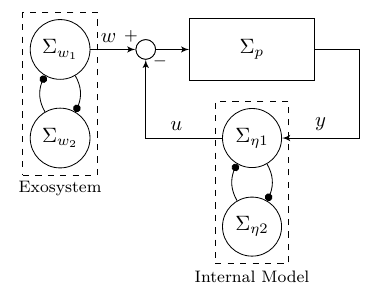}
	\caption{Nonlinear counterpart of the control structure depicted in Figure \ref{fig.control_structure_linear}. The exosystem and controller are now replaced by two nonlinear oscillators.}
	\label{fig.contr_struct.nl}
\end{figure}

The controlled system is modeled as a single neuron system coupled with an excitatory disturbance $w$ and an inhibitory control action $u$, i.e.
\begin{equation}\label{s.E}
	\begin{aligned} 
		\dot{x}_{p_1} &= x_{p_1} - \frac{1}{3}x_{p_1}^3 - x_{p_2} +  k_p \Big ( w - u - y \Big )\\
		\dot{x}_{p_2} &= \tau_p(x_{p_1} - x_{p_2}), \\
		y &= x_{p_1}.
	\end{aligned} 
\end{equation}
Oscllations of the exosystem lead to a disturbance signal $w$ presenting a spike train. As anticipated in the introduction, each spike is associated with an event. As $w$ enters with an excitatory coupling in \eqref{s.E}, if $u=0$ each event provokes a corresponding event in the output $y$ of the controlled system  \eqref{s.E}. The control goal is then to reject all such events by inhibiting them in such a way that the overall input $w-u$ of  \eqref{s.E} is sub-threshold, so as one observes no events in the output $y$.

The design of the controller to reject $w$ is approached by mimicking the linear solution: As in \eqref{s.linear}-\eqref{s.linear.gains}, the controller is designed to embed a copy of the exosystem \eqref{s.w} driven by the output $y(t)$.
In particular, this has the same structure of the system depicted in Figure \ref{fig.contr_struct.nl} and described by the following equations (cf. \eqref{s.w})
\begin{equation}	\label{s.I}
	\begin{aligned}
		\dot{x}_{\eta_1} &= x_{\eta_1} - \frac{1}{3}x^3_{\eta_1} - x_{\eta_2} + k_{\eta}(y -y_{\eta_2}-y_{\eta_1})   \\
		\dot{x}_{\eta_2} &= \tau_{\eta} (x_{\eta 1} - x_{\eta_2}),\\ 
		\dot{x}_{\eta_3} &= x_{\eta_3} - \frac{1}{3}x^3_{\eta_3} - x_{\eta_4} + k_{\eta}(-y_{\eta_1} -y_{\eta_2})  \\
		\dot{x}_{\eta_4} &= \tau_{\eta} (x_{\eta_3} - x_{\eta_4}),\\ 
		y_{\eta_1}  &= x_{\eta_1} ,\quad y_{\eta_2}   = x_{\eta_3} , \quad 
		u  = y_{\eta_1}.
	\end{aligned} 
\end{equation} 
If $(k_\eta, \tau_{\eta}) = (k_w, \tau_w)$, when $y=0$ as desired, then system \eqref{s.I} equals \eqref{s.w}.

Figure~\ref{fig:sim.exact_reg} shows that the controller \eqref{s.I} achieves classical disturbance rejection when the internal model is perfectly calibrated to the exosystem, resulting in a vanishing steady-state output \( y(t) \). In Figure~\ref{fig:sim.below_threshold}, where parametric mismatches between the internal model and the exosystem are introduced, a vanishing steady-state is no longer achieved. Nonetheless, the controller effectively prevents the occurrence of spikes in the system's output, still guaranteeing event rejection.


\section{Discussion}\label{sec.disc}

\subsection{About Closed-loop Contraction}

The top plot of Figure~\ref{fig:sim.not_contr} shows that, in the absence of the disturbance signal $w$, trajectories initialized from different conditions converge to distinct steady-states. This behavior indicates that the closed-loop system \eqref{s.E}--\eqref{s.I} is not contractive. 
However, when the disturbance signal $w$ is introduced, as illustrated in the bottom plot of Figure~\ref{fig:sim.not_contr}, trajectories initialized from different conditions converge to a common steady-state trajectory. This suggests that the system exhibits input-dependent contraction: the presence of $w$ induces convergence across trajectories. This phenomenon is consistent with observations in neuroscience, where neurons driven by suitable stimuli display reliable, synchronized responses across trials, whereas they show incoherent behavior for other stimuli~\citep{mainen1995reliability, brette2003reliability, ermentrout2008reliability}.

\subsection{Similarities with the Linear Case}

Define $\mathcal{F}$ as the class of signals generated by the uncoupled FN model \eqref{eq.fhn} (with $k=0$). The closed-loop system \eqref{s.w}-\eqref{s.I} exhibits the following analogies to the linear regulation framework:
\begin{enumerate}[label=\textbf{F\arabic*.},ref=\textbf{F\arabic*}]
	\item\label{F1} By Lemma~\ref{lem.contractive_FN}, system \eqref{s.E} is convergent, thus BIBO stable, mirroring the internal stability requirement of linear output regulation.
	
	\item\label{F2} Lemma~\ref{lem.contractive_FN} also implies that if the FN neuron \eqref{eq.fhn} is driven by inputs from $\mathcal{F}$, the output converges asymptotically to the input itself. Formally, for inputs satisfying
	\[
	\sum_{i\in \mathcal{I}} \sigma_i z_i \in \mathcal{F},
	\]
	the output $y$ of \eqref{eq.fhn} satisfies
	\[
	\lim_{t\to\infty} \left|y(t)-\sum_{i\in \mathcal{I}}\sigma_i z_i(t)\right|=0,
	\]
	indicating a \emph{pass-all} property relative to $\mathcal{F}$.
	
	\item\label{F3} The closed-loop system explicitly incorporates an internal model that replicates the exosystem, analogous to the internal model principle in the linear framework.
\end{enumerate}

These properties \ref{F1}-\ref{F3} strongly reflect essential concepts from the linear setting. Moreover, leveraging input-dependent contraction observed in the previous subsection, it guarantees a unique steady-state solution compatible with the objective $y=0$.

Interestingly, one might question the necessity of embedding a full internal model identical of the exosystem. After all, the considered model \eqref{eq.fhn} shows sustained oscillations when $k=0$. Thus, could a simpler single-neuron internal model suffice? Consider the following candidate 
\begin{align}\label{eq.single_neuron_candidate}
	\begin{split}
		\dot{x}_{\eta_1} &= x_{\eta_1} - \frac{1}{3} x_{\eta_1}^3 - x_{\eta_2} + k_w(y - x_{\eta_1}), \\
		\dot{x}_{\eta_2} &= \tau(x_{\eta_1} - x_{\eta_2}),
	\end{split}
\end{align}
with output $u = x_{\eta_1}$. Suppose this simpler internal model achieves steady-state rejection, implying $y=0$ at steady-state. Then, at steady-state, system \eqref{eq.single_neuron_candidate} reduces to
\begin{align*}
	\dot{x}_{\eta_1} &= x_{\eta_1} - \frac{1}{3} x_{\eta_1}^3 - x_{\eta_2} - k_w x_{\eta_1},\\
	\dot{x}_{\eta_2} &= \tau(x_{\eta_1} - x_{\eta_2}).
\end{align*}
However, by Lemma \ref{lem.contractive_FN}, this simplified model is contractive, and since $x_{\eta}=0$ is a steady-state solution, it follows that $x_{\eta}(t) \to 0$. This outcome contradicts the requirement for $y=0$, which necessitates $u$ (thus $x_{\eta_1}$) to match the exosystem-generated disturbance $w$. Hence, at least two neurons replicating the exosystem dynamics are necessary for proper internal model implementation.
This necessity for redundancy is coherent with findings in distributed integral control, where sparse network constraints similarly require redundant internal models \citep{Notarnicola2023Gradient, bin2022stability}.

\subsection{Extension to Open Systems}
\begin{figure}[ht]
    \centering
    \includegraphics[width=0.4\textwidth]{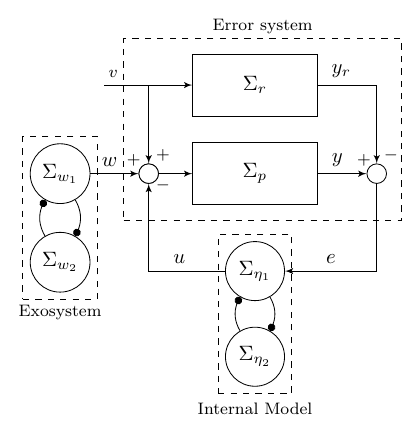}
    \caption{Revisited control structure in the presence of a driving signal $v$. Now the system $\Sigma_p$ is in parallel with a reference system $\Sigma_r$ representing its unperturbed behavior, necessary to build $e$.}
    \label{fig:control_structure_ref}
\end{figure}

This section extends the approach developed in Section~\ref{sec.nl} to the case where the plant is also driven by a measured open input. In this setting, the event disturbance rejection problem can be interpreted as the problem of restoring a desired input-output behavior. This framework encompasses many applications of interest~\citep{bin2023internal, schmetterling2022adaptive}, and introduces an additional degree of freedom through the design of the auxiliary input \( v(t) \). Such a signal plays a twofold role: first, it induces a reliable and reproducible behavior on the controlled system as previously mentioned \citep{kosmidis2003analysis}; second, it specifies the desired input-output spiking pattern that the closed-loop system aims to replicate.

Following the classical output regulation approach, we must construct an error signal \(e(t)\) to solve this problem. This involves the introduction of an auxiliary system $\Sigma_r$ that acts as \emph{reference system}, and is modeled as
\begin{equation}\label{sys.reference}
	\begin{aligned} 
		\dot{x}_{r_1} &= x_{r_1} - \frac{1}{3}x_{r_1}^3 - x_{r_2} +  k_r \Big ( v  - y_r \Big )\\
		\dot{x}_{r_2} &= \tau_r(x_{r_1} - x_{r_2}), \\
		y_r &= x_{r_1},
	\end{aligned} 
\end{equation}
where $(k_r, \tau_r) = (k_p, \tau_p)$. This system is placed in parallel to the plant, generating the error signal \(e(t) := y(t) - y_r(t)\). The resulting closed-loop system is depicted in Figure \ref{fig:control_structure_ref}. Being \eqref{sys.reference} contractive as proven in Lemma \ref{lem.contractive_FN}, also the error system resulting by the parallel interconnection of \eqref{s.E} and \eqref{sys.reference} is contractive \citep{lohmiller1998contraction}. Thus, the same reasoning done before applies also in this case for the error system $\Sigma_e$.

It is important to note that even in this scenario, parametric errors in the knowledge of both the reference and the internal model clearly result in a nonzero steady-state error. However, these modeling errors do not compromise the desired input-output behavior, seen in terms of spiking events in the plant's output, as shown in Figure \ref{fig:sim.with.v}. Also in this case, it seems that the presence of the auxiliary input $v(t)$ helps convergence of the new closed-loop system shown in Figure \ref{fig:control_structure_ref}, following the same reasoning discussed in Section \ref{sec.disc}. 

\section{Conclusions}
This article explored the problem of event disturbance rejection in excitable systems. Inspired by the structure of linear output regulation, the proposed control strategy relies on embedding an internal model and enforcing contraction properties through control design. This enforced convergence allows a steady-state analysis analogous to the linear setting.

The key difference from the classical linear theory is that contraction cannot be guaranteed by system structure alone but must emerge from input-dependent behaviors, a phenomenon closely related to neural reliability observations in neuroscience. By focusing on rejecting discrete events rather than exact trajectories, the proposed design demonstrates robustness against model-controller mismatches: even in the presence of parametric errors, unwanted spikes in the system output are successfully prevented.

Additionally, the extension to open systems was considered, where the additional input signal $v(t)$ contributes to achieving convergence in the closed-loop system. Although perfect regulation of trajectories is lost, the spike pattern is preserved. Also in this case, numerical simulations shows that even in case of parameters mismatch, the spiking behavior remains robust, still preventing unwanted spikes in the system's output.

It is important to note that the presented arguments heavily rely on the considered model \eqref{eq.fhn}, a specific way of modeling synaptic couplings, and the requirement of an internal model for each subproblem to be addressed. Future work will focus on developing purely \emph{event}-oriented coupling mechanisms, resembling the synaptic coupling models commonly used in computational neuroscience. Further analysis will also be conducted on the convergence properties of open systems, with the aim of characterizing the class of input signals that ensure such behavior. Finally, the extent to which these results generalize to other models of neuronal dynamics and synaptic interactions remains an open question and will be the focus of future work.

\begin{figure}
    \centering
    \includegraphics[width=0.425\textwidth]{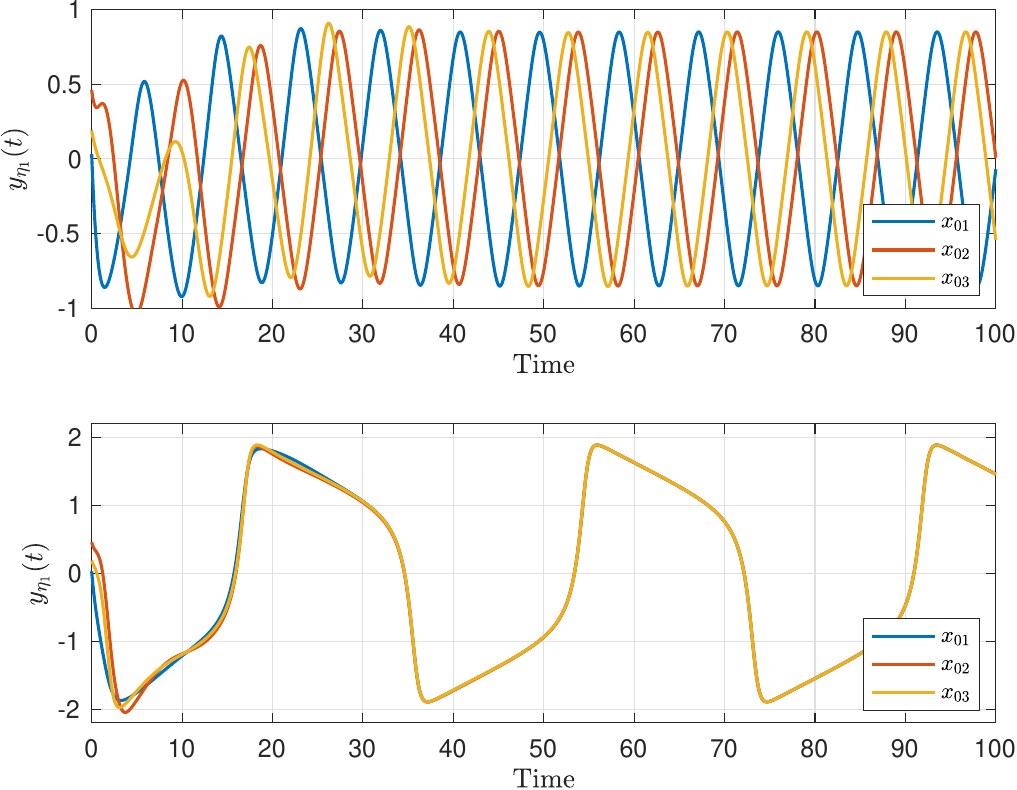}
    \caption{Simulations of the closed-loop system \eqref{s.E}-\eqref{s.I} when $w = 0$ (top plot) and $w \ne 0$ (bottom plot). The system is simulated for different initial conditions, showing that convergence is not achieved  when $w=0$ (top plot), while convergence is enforced when the same signal $w$ is nonzero (bottom plot). The simulations are obtained with $(k_p, \tau_p) = (1, 1/11)$, and $(k_{\eta}, \tau_{\eta})=(2, 1/12)$.}
    \label{fig:sim.not_contr}
\end{figure}
\begin{figure}[h]
    \centering
    \includegraphics[width=0.425\textwidth]{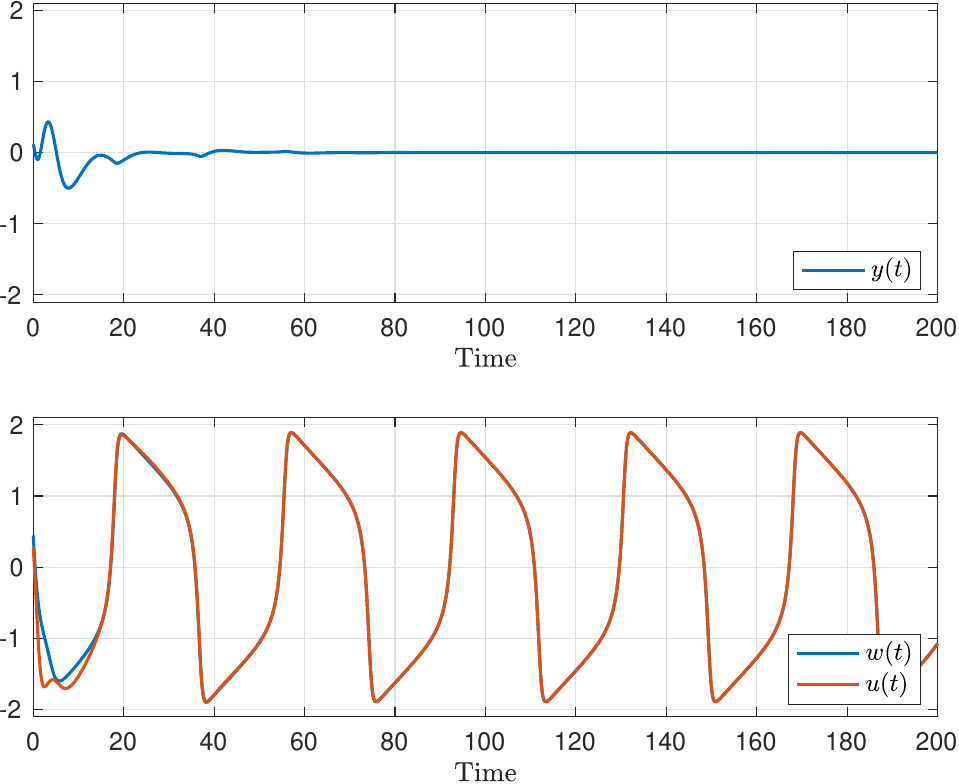}
    \caption{Simulations of the disturbance rejection problem using controller \eqref{s.I}. The internal model synchronizes with the exosystem dynamics (bottom plot) and the plant's output reaches a zero steady-state (top plot). The simulations are obtained with $(k_p, \tau_p) = (1, 1/11)$, and $(k_w, \tau_w) = (k_{\eta}, \tau_{\eta}) = (2, 1/12)$.}
    \label{fig:sim.exact_reg}
\end{figure}

\begin{figure}[h]
    \centering
    \includegraphics[width=0.425\textwidth]{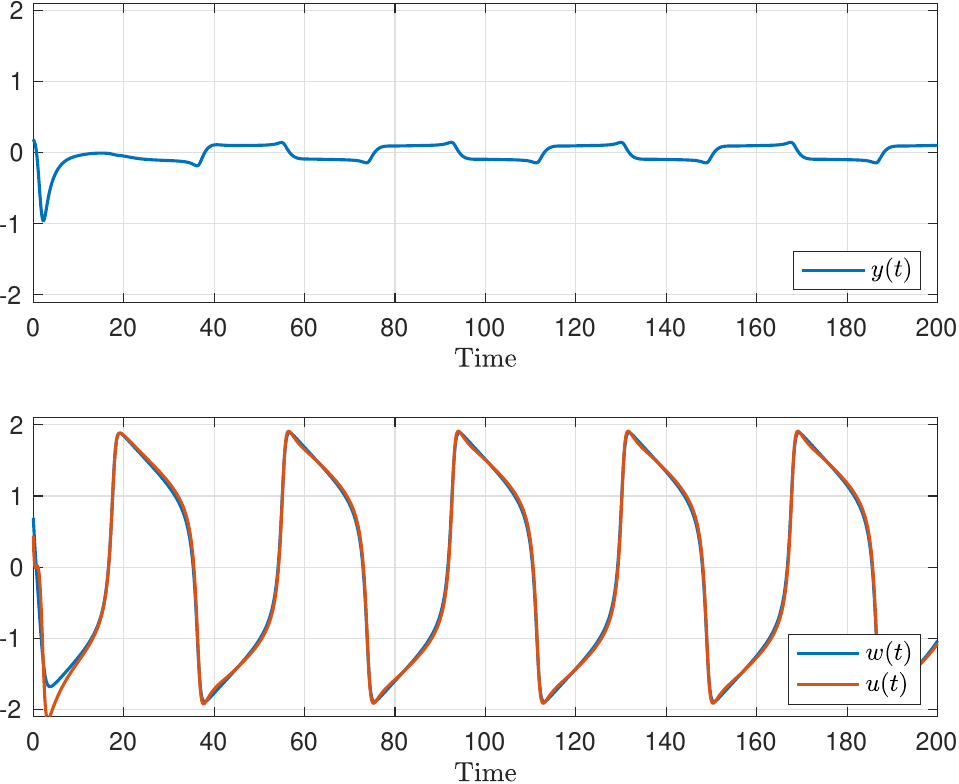}
    \caption{Simulations using controller \eqref{s.I} with $(k_w, \tau_w) \ne (k_{\eta}, \tau_{\eta})$. A small synchronization error between the internal model and the disturbance (bottom plot) generates a non-zero steady-state for the plant's output (top plot). However, while the resulting output oscillates around zero, it does not present any spike in the output. The simulations are obtained with $(k_p, \tau_p) = (1, 1/11)$, $(k_w, \tau_w) = (1 ,1/12)$, and $(k_{\eta}, \tau_{\eta})=(3, 1/16)$.}
    \label{fig:sim.below_threshold}
\end{figure}

\begin{figure}[h]
    \centering
    \includegraphics[width=0.425\textwidth]{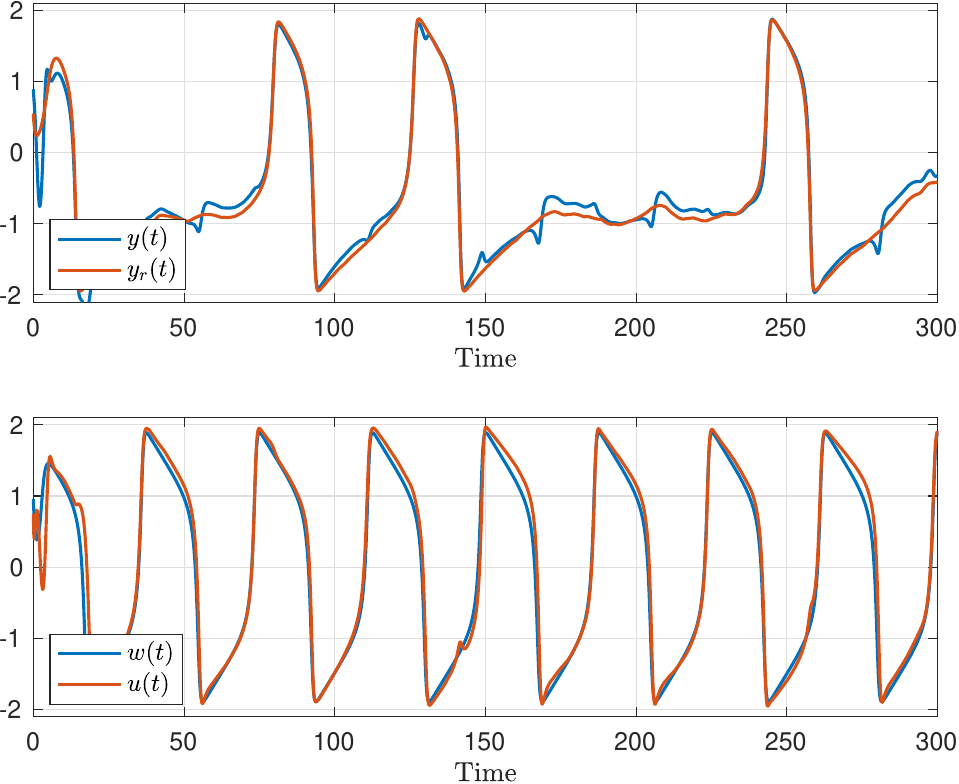}
    \caption{Simulations in the presence of an additional signal $v(t)$, using controller \eqref{s.I} and reference \eqref{sys.reference} with $(k_w, \tau_w) \ne (k_{\eta}, \tau_{\eta})$, and $(k_{r}, \tau_r) \ne (k_p, \tau_r)$. The parametric error in both systems generates a partial synchronization in the internal model (bottom plot) and in the error system (top plot), but does not change the spiking pattern of the system. The simulations are obtained with $(k_p, \tau_p) = (1, 1/11)$, $(k_r, \tau_r) = (3, 1/13)$ $(k_w, \tau_w) = (1, 1/12)$, and $(k_{\eta}, \tau_{\eta})=(4, 1/15)$.}
    \label{fig:sim.with.v}
\end{figure}

\bibliography{bibliography}
\end{document}